\documentclass[letterpaper, 10 pt, conference]{ieeeconf}  %

\IEEEoverridecommandlockouts                              % This command is only needed if you want to use the \thanks command

\let\labelindent\relax

\usepackage{cite}
\usepackage{amssymb,amsfonts, mathtools}
\usepackage[table]{xcolor}
\usepackage{bm}
\usepackage{enumitem}
\usepackage{siunitx}
\usepackage{makecell}
\usepackage{graphicx}
\usepackage[colorlinks,
  linkcolor=blue,
  citecolor=blue, urlcolor=.]{hyperref}
\usepackage[capitalise]{cleveref}
\crefname{assumption}{Assumption}{Assumptions}
\crefname{enumi}{Assumption}{}
\usepackage{amsthm}

\newtheoremstyle{ieeeconf}
  {0pt}   % ABOVESPACE
  {0pt}   % BELOWSPACE
  {\normalfont}  % BODYFONT
  {\parindent}       % INDENT (empty value is the same as 0pt)
  {\itshape} % HEADFONT
  {:}         % HEADPUNCT
  { } % HEADSPACE
  {\thmname{#1} \thmnumber{#2}\thmnote{ (#3)}} % CUSTOM-HEAD-SPEC
\makeatletter
\renewenvironment{proof}[1][\proofname]{\par
  \pushQED{\qed}%
  \normalfont \topsep\z@
  \trivlist
  \item[\hskip2em
        \itshape
    #1\@addpunct{:}]\ignorespaces
}{%
  \popQED\endtrivlist\@endpefalse
}
\makeatletter

\theoremstyle{ieeeconf}

\newtheorem{definition}{Definition}
\newtheorem{lemma}{Lemma}
\newtheorem{proposition}{Proposition}
\newtheorem{remark}{Remark}
\newtheorem{theorem}{Theorem}
\newtheorem{corollary}{Corollary}

\DeclareMathOperator{\Tr}{Tr}

\newcommand{\inp}{\boldsymbol{u}}

\newcommand{\R}{\mathbb{R}}

\newcommand{\Prob}{\mathbb{P}}
\newcommand{\E}{\mathbb{E}}
\newcommand{\Q}{\mathbb{Q}}

\newcommand{\Symm}{\mathbb{S}}

\newcommand{\dd}{\mathop{}\!\mathrm{d}}

\newcounter{relctr} %% <- counter for relations
\everydisplay\expandafter{\the\everydisplay\setcounter{relctr}{0}} %% <- reset every eq
\AtBeginDocument{} %% <- store original definition

\title{\LARGE \bf On the Global Optimality of Linear Policies for \\Sinkhorn Distributionally Robust Linear Quadratic Control
}

\author{Riccardo Cescon, Andrea Martin, and Giancarlo Ferrari-Trecate
\thanks{R. Cescon and G. Ferrari-Trecate
are with the Institute of Mechanical Engineering, EPFL, Switzerland. E-mail addresses: \{riccardo.cescon, giancarlo.ferraritrecate\}@epfl.ch.}
\thanks{A. Martin is with the School of
Electrical Engineering and Computer Science, and Digital Futures, KTH Royal Institute of Technology, Sweden. E-mail address: andrmar@kth.se.}
\thanks{This work was supported as a part of NCCR Automation, a National Centre of Competence in Research, funded by the Swiss National Science Foundation (grant number 51NF40\_225155), and by Digital Futures.}
}

\begin{document}

\maketitle

\thispagestyle{empty}
\pagestyle{empty}

\begin{abstract}
% The LQG is a cornerstone in control theory with a wide range of applications in which it has been used. Despite its simple nature, it was shown that when the model is not known exactly and the noise affecting the system is not Gaussian the LGQ provides no robustness guarantees leading to unexpected behavior and poor performances. This issue has been tackled in different ways, trying to find a policy able to retain stability in presence of uncertainty. In recent years, one approach to guarantee performances even when the noise is not known is Distributionally Robust Control. Recent works showed that when the uncertainty is modeled with a Wasserstein or relative entropy ambiguity set with normal center the optimal policy minimizing the worst-case cost is still linear and the worst-case distribution is Gaussian. In line with this works, our paper establishes a novel bound for the Sinkhorn discrepancy; it shows that the entropy-regularized ambiguity set is convex and compact and finally uses such results to show that even in our setup the linear policies are globally optimal.
% The Linear Quadratic Gaussian (LQG) controller is a cornerstone of control theory with broad applicability. However, its performance degrades significantly when the system model is uncertain or the process noise deviates from the Gaussian assumption. In such cases, the LQG controller offers no robustness guarantees, often resulting in unstable behavior and poor performance.
The Linear Quadratic Gaussian (LQG) regulator is a cornerstone of optimal control theory, yet its performance can degrade significantly when the noise distributions deviate from the assumed Gaussian model. To address this limitation, this work proposes a distributionally robust generalization of the finite-horizon LQG control problem. Specifically, we assume that the noise distributions are unknown and belong to ambiguity sets defined in terms of a Sinkhorn discrepancy centered at a nominal Gaussian distribution. By deriving novel bounds on this entropy-regularized Wasserstein distance and proving structural and topological properties of the resulting ambiguity sets, we establish global optimality of linear policies for Sinkhorn distributionally robust LQG. Numerical experiments showcase improved distributional robustness of our control policy.

\end{abstract}
\section{Introduction}
The theory of Linear Quadratic Gaussian (LQG) regulators addresses the fundamental problem of controlling partially-observed linear systems driven by additive Gaussian noise with the objective of minimizing an expected quadratic cost \cite{bertsekas2012dynamic}. This problem admits an elegant closed-form solution, combining a Kalman filter with a linear state-feedback controller, and has found application in a variety of domains ranging from engineering to economics and computer science.

In the presence of model misspecifications, however, the LQG solution can be extremely fragile \cite{doyle1978guaranteed}. Classical $\mathcal{H}_\infty$ control \cite{zhou1998essentials} addresses this concern by shifting from a stochastic to an adversarial uncertainty model and minimizing the worst-case cost across bounded-energy disturbances. While provably robust, $\mathcal{H}_\infty$ methods tend to be overly conservative in practice as they optimize for the least favorable uncertainty realization. Motivated by this observation, several approaches have been proposed to balance nominal performance and robustness, including mixed $\mathcal{H}_2/\mathcal{H}_\infty$ formulations \cite{bernstein1988lqg, doyle1989optimal}, risk-sensitive control \cite{jacobson1973optimal}, and regret minimization methods \cite{martin2025guarantees, goel2023regret, martin2024regret}.

Among these approaches, recent work on distributionally robust (DR) control promises to combine the advantages of stochastic and adversarial uncertainty models by \emph{robustifying in the probability space}. To achieve this, the DR control paradigm considers the problem of minimizing the expected cost under the most averse distribution within a given ambiguity set—a set of distributions that are sufficiently close, in an appropriate sense, to a nominal one. For instance, \cite{van2015distributionally} studied DR control of constrained stochastic systems with ambiguity sets comprising all distributions sharing the same first two moments. To account for full distributional information, the works \cite{petersen2000minimax, falconi2025distributionally, fochesato2025distributionally} instead employ ambiguity sets defined in terms of $f$-divergence, whereas \cite{taskesen2023distributionally, lanzetti2024optimality, yang2020wasserstein, brouillon2025distributionally, aolaritei2023wasserstein, kargin2024infinite} rely on optimal transport metrics in light of their proven expressiveness and out-of-distribution guarantees.\looseness-1

Inspired by these results and motivated by the recent application of the Sinkhorn discrepancy for DR learning \cite{sinkhorn, genevay2018learning} and control \cite{cescon2025data}, we study a generalization of the finite-horizon LQG control problem, where the noise distributions are unknown and belong to Sinkhorn ambiguity sets centered at nominal Gaussian distributions. Our main contribution is to establish the global optimality of linear policies for the Sinkhorn DR LQG problem. While our proof follows a ``sandwich'' argument akin to \cite{taskesen2023distributionally}, the extension from Wasserstein DR LQG to Sinkhorn DR LQG is far from direct: the introduction of the KL regularization term fundamentally alters the geometry of the ambiguity set\footnote{For instance, it is well-known that the worst-case distribution in Wasserstein DR optimization is finitely supported when so is the center distribution \cite{gao2023distributionally}, while the Sinkhorn ambiguity set with a Gaussian reference measure only contains continuous distributions.} and invalidates several key inequalities used in \cite{taskesen2023distributionally}. To recover a comparable optimality result, we develop new tools specific to the Sinkhorn framework. In particular,
we first construct a novel \emph{lower bound} to our DR LQG problem leveraging results from regularized optimal transport \cite{del2020statistical};
second, we establish a Gelbrich-type inequality for the Sinkhorn discrepancy, which is not available in the literature and is essential for deriving an \emph{upper bound} to our DR LQG problem; last, we prove convexity and compactness of the resulting entropy-regularized Gelbrich ambiguity set. Taken together, these results allow us to conclude that the Sinkhorn DR LQG admits a Nash equilibrium in the class of linear policies.

%, thereby extending the result of \cite{taskesen2023distributionally} beyond the classical Wasserstein setting to the Sinkhorn discrepancy \cite{cuturi2013sinkhorn}

% Our main contribution is to establish global optimality of linear policies for this Sinkhorn DR LQG control problem, generalizing the results of \cite{taskesen2023distributionally} to the case where the definition of Wasserstein distance includes a Kullback-Leibler (KL) regularization term \cite{cuturi2013sinkhorn}. Towards deriving our main result, we first construct a \emph{lower bound} to our DR LQG problem leveraging results from regularized optimal transport \cite{del2020statistical}. Second, we prove a Gelbrich-type inequality that bounds from below the Sinkhorn discrepancy between two probability distributions when one of them is Gaussian. This allows us to obtain an \emph{upper bound} to our DR LQG problem. Third, we show convexity and compactness of the resulting entropy-regularized Gelbrich ambiguity set. These properties allow us to conclude, using a ``sandwich'' argument similar to \cite{taskesen2023distributionally}, that the Sinkhorn DR LQG admits a Nash equilibrium in the class of linear policies.

Alongside \cite{taskesen2023distributionally, lanzetti2024optimality, fochesato2025distributionally, falconi2025distributionally}, our work contributes to delineating scenarios where, despite the additional complexity introduced by DR formulations, linear feedback policies remain globally optimal for LQG control problems.

\textit{Notation}: Throughout the paper, we denote the set of (zero-mean) probability distributions supported on a measurable set $\mathcal{Z}$ by $\mathcal{P}(\mathcal{Z})$ (resp. $\mathcal{P}_0(\mathcal{Z})$). We write $\mu \ll \nu$ to denote that a measure $\mu$ is absolutely continuous with respect to $\nu$. The convolution product between two probability measures is represented by $\mu*\nu$. For $n \in \mathbb{N}$, we write $[n]$ to denote the set of indices $\{0,\dots, n-1\}$. The space of all $d\times d$ positive (semi)definite matrices is denoted by \( \mathbb{S}_{++}^d \) (resp. \( \mathbb{S}^d_+ \)). We denote by $\|\cdot\|$ the Euclidean norm. The determinant of a square matrix $A$ is denoted by $|A|$. Given $A \in \Symm^{d}$, we denote by $\{\lambda_i(A)\}_{i = 1}^d$ its eigenvalues and let $\lambda_{\text{max}}(A) = \max_i \lambda_i(A)$.% its largest one.\looseness=-1

\section{Preliminaries}
% This section recalls regularized transport discrepancies that will be used in the sequel. The Sinkhorn discrepancy adds to the Wasserstein distance a Kullback–Leibler penalty on the coupling. In the Gaussian case this yields to a closed-form expression termed entropy-regularized Gelbrich divergence.
We begin by recalling definitions of discrepancies between probability distributions that will be used throughout the paper.
\begin{definition}[KL divergence, {\cite[Definition 2.8]{kuhn2025distributionally}}]
Given $\Prob, \Q\in\mathcal{P}(\R^d)$ with $\Prob\ll \Q$, the Kullback-Leibler (KL) divergence between $\Prob$ and $\Q$ is
    \begin{equation*}
    \mathrm{KL}(\Prob\|\Q) = \E_\Prob\left[\log\left(\frac{\dd\Prob(x)}{\dd\Q(x)}\right)\right].
\end{equation*}
\end{definition}
\begin{definition}[Sinkhorn discrepancy, {\cite[Definition 1]{sinkhorn}}]
Let $\Prob, \Q\in\mathcal{P}(\R^d)$ and $\mu, \nu$ be reference probability measures over $\R^d$ 
such that $\Prob \ll \mu$ and $\Q \ll \nu$. For any $\epsilon \geq 0$, the Sinkhorn discrepancy 
between $\Prob$ and $\Q$ is defined as
\begin{equation}
\label{eq:sinkhorn}
    W_{\epsilon}(\Prob, \Q) = \! \inf_{\gamma \in \Gamma(\Prob, \Q)}\left\{\E_\gamma [\|x-y\|^2] + \\ 
    \epsilon \mathrm{KL}(\gamma\|\mu\times\nu )\right\},
\end{equation}
where $\Gamma(\Prob, \Q)$ denotes the set of all couplings $\gamma$ between $\Prob$ and $\Q$, that is, the set of all joint distributions with marginals $\Prob$ and $\Q$.
\end{definition}
As noted in \cite[Remark 2]{sinkhorn}, any choice of $\Prob\ll\mu$ in \eqref{eq:sinkhorn} is equivalent up to a constant. Hence, as in our DR LQG problem the center distribution is assumed to be known, we let $\mu=\Prob$ without loss of generality. Aligned with \cite{sinkhorn, cescon2025data}, in the following we further assume $\nu\sim\mathcal{N}(0, \Sigma)$; this choice ensures that $W_{\epsilon}(\Prob, \Q)$ is finite for any $\Q$ in the ambiguity set $\{\Q \in \mathcal{P}_0(\mathbb{R}^d): W_\epsilon(\Prob, \Q) \leq \rho\}$, independently of its support.

\begin{remark}
\label{remark:equivalence}
    Other definitions of Sinkhorn discrepancy have also been considered in the literature \cite{del2020statistical, janati2020entropic, mallasto2022entropy}. For instance, \cite{del2020statistical} regularizes the transport cost with the negative differential entropy of the transport plan $\gamma$, while \cite{janati2020entropic, mallasto2022entropy} use \eqref{eq:sinkhorn} with $\mu=\Prob$ and $\nu=\Q$. Crucially, all these definitions lead to the same optimal transport plan $\gamma^\star$. In fact, one can observe that 
    \begin{equation*}
        \mathrm{KL}(\gamma\|\mu\times\nu) = \mathrm{KL}(\gamma\|\Prob\times\Q) + \mathrm{KL}(\Prob\times\Q\|\mu\times\nu)\,;
    \end{equation*}
    hence, $\mathrm{KL}(\gamma\|\mu\times\nu)$ and $\mathrm{KL}(\gamma\|\Prob\times\Q)$ are equivalent up to a term that is independent of $\gamma$. A similar reasoning applies when the negative differential entropy is used as regularization term \cite{del2020statistical}.   
\end{remark}

\begin{definition}
\label{def:entropy_reg_gelbrich}
    Let $\nu\sim\mathcal{N}(0, \Sigma)$ where $\Sigma \in \Symm^d_{++}$. The entropy-regularized Gelbrich divergence between two probabilities $\mathbb{P}_1, \mathbb{P}_2 \in \mathcal{P}_0(\mathbb{R}^d)$ with covariance matrices $\Sigma_1, \Sigma_2\in\Symm_{++}^d$ is
    % \begin{equation*}
    %     G_\epsilon(\Sigma_1, \Sigma_2) = \Tr(\Sigma_1) + \Tr(\Sigma_2) -2 \Tr(D_\epsilon)+\frac{\epsilon}{2}\Tr\left(\Sigma^{-1}\Sigma_2\right) + 
    %     \frac{\epsilon}{2}\log\frac{|\Sigma|}{|\Sigma_2|} +\frac{\epsilon}{2}\log\left(\left(\frac{2}{\epsilon}\right)^d\left|D_\epsilon + \frac{\epsilon}{4}I\right|\right),
    % \end{equation*}
    \begin{align}
    G_\epsilon(\Sigma_1, \Sigma_2)\! &=\! \Tr(\Sigma_1)\! +\! \Tr(\Sigma_2)\! -\! 2\Tr(D_\epsilon) 
    \!+\! \frac{\epsilon}{2}\Tr\left(\Sigma^{-1}\Sigma_2\right)\nonumber \\
    &+ \frac{\epsilon}{2}\log\frac{|\Sigma|}{|\Sigma_2|} 
    \!+\! \frac{\epsilon}{2}\log\left(\left(\frac{2}{\epsilon}\right)^d \left|D_\epsilon + \frac{\epsilon}{4}I\right|\right),\label{eq:entropy_reg_gelbrich}
    \end{align}
where $D_\epsilon = \left(\Sigma_1^\frac{1}{2}\Sigma_2\Sigma_1^\frac{1}{2} + \frac{\epsilon^2}{16}I\right)^\frac{1}{2}$.
\end{definition}

Unlike \eqref{eq:sinkhorn}, \cref{def:entropy_reg_gelbrich} only accounts for the covariances of $\mathbb{P}_1, \mathbb{P}_2 \in \mathcal{P}_0(\mathbb{R}^d)$. In \cref{sec:analysis}, we will exploit connections between \eqref{eq:sinkhorn} and \eqref{eq:entropy_reg_gelbrich} to construct finite-dimensional upper and lower bounds to our Sinkhorn DR LQG problem.
% \begin{remark}
    % Note that our definition reduces to the expression in \cite[Theorem 1]{janati2020entropic} when the regularization parameter is set to $\epsilon = 2\sigma^2$ and the reference distribution $\nu$ coincides with the marginal $\Q$. This corresponds to $\Sigma = \Sigma_2$, for which $\Tr(\Sigma^{-1}\Sigma_2) = d$ and $\log\frac{|\Sigma|}{|\Sigma_2|} = 0$, causing the additional terms in our expression to vanish. 
    
% \end{remark}
% We conclude this section highlighting that \eqref{eq:sinkhorn} is not a metric and the minimum over $\Q$ is attained by $\Prob * \mathcal{N}(0, \frac{\epsilon}{2}I)$ \cite[Theorem 2.4]{del2020statistical}. Remarkably, said minimum is not zero and it is not attained by $\Prob$. In the case of $\Prob\sim\mathcal{N}(0, \Sigma_1)$, its value is given by $\underline{\rho} = \frac{\epsilon}{2}\left(\Tr\left(\Sigma^{-1}(\Sigma_1 + \frac{\epsilon}{2}I)\right) - d + \log|\Sigma| - d\log(\frac{\epsilon}{2})\right)$. Such value is the smallest radius to guarantee that the ambiguity sets in the following are non-empty.
We conclude this section by observing that \eqref{eq:sinkhorn} does not define a metric, as it does not satisfy the identity of indiscernibles. In fact, the minimum of $W_\epsilon(\Prob, \Q)$ over $\Q$ is attained at $\Prob * \mathcal{N}(0, \tfrac{\epsilon}{2}I)$, as shown in \cite[Theorem 2.4]{del2020statistical}. Interestingly, this minimum is non-zero and is not achieved by $\Prob$ itself. In the case $\Prob \sim \mathcal{N}(0,\hat{\Sigma})$ that we consider throughout the paper, the minimum of \eqref{eq:sinkhorn} over $\Q$ is given by \looseness=-1
\begin{equation}
\label{eq:smallest_radius}
    \underline{\rho} = \tfrac{\epsilon}{2}\big(\Tr\!\big(\Sigma^{-1}(\hat{\Sigma} + \tfrac{\epsilon}{2}I)\big) - d + \log|\Sigma| - d\log(\tfrac{\epsilon}{2})\big)\,.
\end{equation}
This value represents the smallest radius $\rho$ such that the ambiguity set $\{\Q \in \mathcal{P}_0(\mathbb{R}^d): W_\epsilon(\Prob, \Q) \leq \rho\}$ is non-empty. Last, we remark that, when $\epsilon\rightarrow 0$, Definition~\ref{def:entropy_reg_gelbrich} is equivalent to the squared Gelbrich distance, see \cite[Theorem 2.1]{gelbrich1990formula}.
\section{Problem Formulation}
We consider discrete-time linear dynamical systems described by the following state-space equations
\begin{equation}
\label{eq:model}
x_{t+1} =A_t x_{t}+B_t u_{t}+w_{t}\,, \quad 
y_{t} =C_t x_{t}+v_{t}\,, ~ \forall t \in [T]\,,
\end{equation}
where $x_t\in\R^d$ denotes the state vector, $u_t\in\R^m$ the control input, $y_t\in\R^p$ the output, $w_t\in\R^d$ the process noise, $v_t\in\R^p$ the measurement noise, and $T \in \mathbb{N}$ the control horizon. 

For ease of presentation, we collect all exogenous random vectors in the variable $\delta = (x_0, w_0, \dots, w_{T-1}, v_0, \dots, v_{T-1})$. We assume that all entries of $\delta$ are mutually independent.
%We assume that the exogenous random vectors $x_0$, $\{w_t\}_{t=0}^{T-1}$ and $\{v_t\}_{t=0}^{T-1}$ are mutually independent. % and distributed accordingly to $\Prob^\star_{x_0}$, $\{\Prob_{w_t}^\star\}_{t=0}^{T-1}$ and $\{\Prob_{v_t}^\star\}_{t=0}^{T-1}$, respectively. 
Differently from classical LQG theory, however, we assume that their true distributions are unknown and belong to a Sinkhorn ambiguity set $\mathcal{S}$ centered at a nominal distribution $\hat{\Prob} = \hat{\Prob}_{x_0} \otimes \left(\otimes_{t=0}^{T-1} \hat{\Prob}_{w_t}\right) \otimes \left(\otimes_{t=0}^{T-1} \hat{\Prob}_{v_t}\right)$, with $\hat{\Prob}_{x_0}=\mathcal{N}(0, \hat{X}_0)$, $\hat{\Prob}_{w_t}=\mathcal{N}(0, \hat{W}_t)$, and $\hat{\Prob}_{v_t}=\mathcal{N}(0, \hat{V}_t)$, for all $t\in[T]$.\footnote{Following \cite{taskesen2023distributionally, lanzetti2024optimality, fochesato2025distributionally, falconi2025distributionally}, we consider only distributions with zero mean; our results extend to the non-zero mean case with minor modifications.} Specifically, for any regularization parameter $\epsilon\geq0$ and user-defined radii $\rho_{x_0}\geq \underline{\rho}_{x_0}, \rho_{w_t}\geq \underline{\rho}_{w_t}, \rho_{v_t}\geq \underline{\rho}_{v_t}$, where $\underline{\rho}_{x_0}, \underline{\rho}_{w_t}$, and $\underline{\rho}_{v_t}$ are defined according to \eqref{eq:smallest_radius} with $\hat{X}_0$, $\hat{W}_t$, and $\hat{V}_t$ in place of $\hat{\Sigma}$, respectively, we define the ambiguity set $\mathcal{S}$ as $\mathcal{S}_{x_0}\otimes \left(\otimes_{t=0}^{T-1} \mathcal{S}_{w_t}\right) \otimes \left(\otimes_{t=0}^{T-1} \mathcal{S}_{v_t}\right)$
where
\begin{equation*}
    \begin{aligned}
        \mathcal{S}_{x_{0}}&=\left\{\mathbb{P}_{x_{0}} \in \mathcal{P}_0(\mathbb{R}^{d}): W_{\epsilon}(\hat{\mathbb{P}}_{x_{0}}, \mathbb{P}_{x_{0}}) \leq \rho_{x_{0}}\right\}\,,
        \\
        \mathcal{S}_{w_{t}}&=\left\{\mathbb{P}_{w_{t}} \in \mathcal{P}_0(\mathbb{R}^{d}): W_{\epsilon}(\hat{\mathbb{P}}_{w_{t}}, \mathbb{P}_{w_{t}}) \leq \rho_{w_{t}}\right\}\,,
        \\
        \mathcal{S}_{v_{t}}&=\left\{\mathbb{P}_{v_{t}} \in \mathcal{P}_0(\mathbb{R}^{p}): W_{\epsilon}(\hat{\mathbb{P}}_{v_{t}},\mathbb{P}_{v_{t}}) \leq \rho_{v_{t}}\right\}\,.
    \end{aligned}
\end{equation*}
Given a realization $\delta$ of the exogenous vectors and a collection of causal measurable\footnote{Throughout the paper, we tacitly assume that the probability space of the exogenous signals is equipped with the standard Borel $\sigma$-algebra.} functions $\pi_t:\R^{p(t+1)}\rightarrow\R^m$ mapping output observations to control inputs as per $u_t = \pi_t(y_{0:t})$, we define the cost incurred by the policy $\pi = (\pi_0, \dots, \pi_{T-1})$ as%the exogenous vectors  $x_0$, $\{w_t\}_{t=0}^{T-1}$ and $\{v_t\}_{t=0}^{T-1}$, as 
\begin{equation*}
    J(\pi, \delta) = \sum_{t=0}^{T-1} (x_t^\top Q_t x_t + u_t^\top R_t u_t) + x_T^\top Q_T x_T\,,
\end{equation*}
where $Q_t, Q_T \in \Symm^d_+$ and $R_t\in\Symm^m_{++}$ represent state and input weight matrices, respectively. With these definitions in place, we formulate the Sinkhorn DR LQG control problem as
\begin{equation}
\label{eq:DR-LQG}
    \inf _{\pi\in\mathcal{U}_y}\ \max_{\mathbb{P} \in \mathcal{S}} \ \mathbb{E}_{\Prob}[J(\pi, \delta)]\,,
\end{equation}
where $\mathcal{U}_y$ denotes the set of all feasible control inputs $ \inp =(u_0, \dots, u_{T-1})$. In particular, \eqref{eq:DR-LQG} can be interpreted as a zero-sum game between the control designer, who selects a causal policy to minimize the expected cost, and an adversary, who chooses the noise distributions in $\mathcal{S}$ that maximize such cost. %As a result,  the control designer aims at finding the causal policy minimizing the worst-case expected cost across all admissible distributions in the Sinkhorn ambiguity set $\mathcal{S}$.

\section{Analysis of the Sinkhorn DR LQG problem}
\label{sec:analysis}
We now present our main results. We first re-parametrize \eqref{eq:DR-LQG} in terms of the purified observations \cite[Section 14.4.2]{ben2009robust}. Then, we construct two auxiliary problems over the class of linear policies that provide lower and upper bounds to \eqref{eq:DR-LQG}. Last, using a ``sandwich" argument, we show that the optimal value of these two auxiliary problems coincide, implying that \eqref{eq:DR-LQG} admits a Nash equilibrium and that linear policies are globally optimal.\looseness-1

\subsection{Purified output re-parametrization}
\label{sec:purified-output}
To ease analysis of \eqref{eq:DR-LQG}, we rewrite $u_t$ in terms of the purified observations instead of the actual observations $y_{0:t}$. To define these new variables, we introduce a noise-free copy of \eqref{eq:model} as
\begin{equation}
    \hat{x}_{t+1} = A_t \hat{x}_{t}+B_t u_{t},\quad \hat{y}_{t} = C_t \hat{x}_{t},~ \forall t\in[T],
    \label{eq:noise-free}
\end{equation}
with state $\hat{x}_t\in\R^d$ and output $\hat{y}_t\in\R^p$, initialized at $\hat{x}_0=0$ and with the same input $u_t$ as the 
original system. We then define the \emph{purified output} $\eta_t$ at time \( t \) as \( \eta_t = y_t - \hat{y}_t \), and let \( \bm{\eta} = (\eta_0, \dots, \eta_{T-1}) \). This representation proves useful because, as shown in 
 \cite[Proposition II.1]{hadjiyiannis2011efficient}, every measurable function of \( y_0, \dots, y_t \) can be equivalently expressed as a measurable function of \( \eta_0, \dots, \eta_t \), and vice versa—yet differently from $\mathbf{y} = (y_0, \dots, y_{T-1})$, \(\bm\eta\) is independent of the inputs. In fact, using \eqref{eq:model} and \eqref{eq:noise-free} recursively, one can show that \(\bm\eta\) only depends on the exogenous vectors $\mathbf{w} = (x_0, w_0, \dots, 
w_{T-1})$ and $\mathbf{v} = (v_0, \dots, v_{T-1})$. In particular, it holds that $\bf \bm{\eta} = D w + v$, where $\bf D = CG$ and $\bf C$ and $\bf G$ are matrices defined in the Appendix. 

In light of this, we have that \( \mathcal{U}_y = \mathcal{U}_\eta \), where \( \mathcal{U}_\eta \) denotes the set of input sequences \( \mathbf{u} = (u_0, \dots, u_{T-1}) \) for which there exist measurable functions \( \tilde{\pi}_t: \mathbb{R}^{p(t+1)} \to \mathbb{R}^m \) satisfying \( u_t = \tilde{\pi}_t(\eta_{0:t}) \). Hence, we equivalently rewrite problem \eqref{eq:DR-LQG} as
\begin{equation}
    \label{eq:primal}
    p^\star =
    \left\{
    \begin{aligned}
    &\min_{\bf u}\ \max_{\mathbb{P} \in \mathcal{S}}\ \mathbb{E}_{\mathbb{P}}\left[\bf u^{\top} R u + x^{\top} Q x \right]
    \\
    &\text { s.t. }\ \mathbf{u} \in \mathcal{U}_{\eta}, \quad \bf x = H u + G w,
    \end{aligned}
    \right.
\end{equation}
where $\mathbf{x} = (x_0, \dots, x_T)$ and $\bf R, Q, H$ are suitable matrices defined in the Appendix.

% This reformulation is useful because the purified outputs $\bm\eta$ are independent of the inputs. Indeed, one can show using recursively the equations 
% % of the original system combined with the noise-free version that the purified observations can be expressed as a linear combination of the exogenous 
% uncertainties. Specifically, it holds that $\bf \bm{\eta} = D w + v$, where $\mathbf{v} = (v_0, \dots, v_{T-1})$, $\bf D = CG$, and $\bf C$ defined in \cref{appendix}.

With the objective of establishing the existence of a Nash equilibrium, we also define the dual problem of \eqref{eq:primal} as
\begin{equation}
    \label{eq:dual}
    d^\star =
    \left\{
    \begin{aligned}
    &\max_{\mathbb{P} \in \mathcal{S}}\ \min_{\bf u}\ \mathbb{E}_{\mathbb{P}}\left[\bf u^{\top} R u + x^{\top} Q x \right]
    \\
    &\text { s.t. }\ \mathbf{u} \in \mathcal{U}_{\eta}, \quad \bf x = H u + G w.
    \end{aligned}
    \right.
\end{equation}
Classical min-max inequality states that $d^\star \leq p^\star$. In the next sections, we prove that such relation actually holds with equality—despite the fact that $\mathcal{U}_\eta$ is an infinite-dimensional function space and $\mathcal{S}$ is an infinite-dimensional set of non-parametric probability distributions. In particular, our analysis reveals that there exists a Nash equilibrium of the zero-sum game \eqref{eq:DR-LQG} in the form $(\mathbf{u}^\star, \Prob^\star)$, where $\bf u^\star = U^\star\bm\eta + q^\star$ is an affine policy of $\bm\eta$ and $\Prob^\star$ is Gaussian. To do so, we first construct a lower-bound to \eqref{eq:dual} and an upper-bound to \eqref{eq:primal}, and then show that their optimal values coincide.

\subsection{Construction of a lower bound for \texorpdfstring{$d^\star$}{d*}}
% To derive a lower bound for the dual problem \eqref{eq:dual}, we intersect the set \( \mathcal{S} \) with the family of Gaussian distributions and denote this restricted set by \( \mathcal{S}_{\mathcal{N}} \). The resulting bounding problem is obtained as
To derive a lower bound for the dual problem \eqref{eq:dual}, we restrict our attention to the set \( \mathcal{S}_{\mathcal{N}} \subseteq \mathcal{S} \) of Gaussian distributions contained in \( \mathcal{S} \). This leads us to the optimization problem
%and denote this subset by \( \mathcal{S}_{\mathcal{N}} \). 
\begin{equation}
    \label{eq:lower_bound}
    \underline{d}^\star =
    \left\{
    \begin{aligned}
    &\max_{\mathbb{P} \in \mathcal{S_N}}\ \min_{\bf u}\ \mathbb{E}_{\mathbb{P}}\left[\bf u^{\top} R u + x^{\top} Q x \right]
    \\
    &\text { s.t. }\ \mathbf{u} \in \mathcal{U}_{\eta}, \quad \bf x = H u + G w.
    \end{aligned}
    \right.
\end{equation}
Compared to \eqref{eq:dual}, in \eqref{eq:lower_bound} we restricted the feasible set in the outer maximization. Hence, we have $\underline{d}^\star \leq d^\star$. On the other hand, \eqref{eq:lower_bound} is still an infinite-dimensional optimization problem. In the remainder of this section, we show that this problem can be reformulated as a finite-dimensional one by exploiting known closed-form expressions for the Sinkhorn discrepancy between two normal distributions.
%We can notice that \eqref{eq:lower_bound} is an infinite-dimensional optimization problem. Nevertheless, we will show that it can be recast into a finite-dimensional one exploiting a known-fact for the Sinkhorn discrepancy between two normal distributions as stated in the following proposition.
\begin{proposition}(Tightness for normal distributions). 
For any \( \mathbb{P}_1 \sim \mathcal{N}(0, \Sigma_1) \) and \( \mathbb{P}_2 \sim \mathcal{N}(0, \Sigma_2) \) with $\Sigma_1, \Sigma_2 \in \Symm_{++}^d$, the optimal coupling for the entropy-regularized problem \eqref{eq:sinkhorn} is Gaussian and is given by
$\gamma_0 \sim \mathcal{N}\left(0, \begin{bmatrix} 
\Sigma_1 & \Sigma_1 X_\epsilon \\
X_\epsilon \Sigma_1 & \Sigma_2
\end{bmatrix} \right)$,
where
\begin{equation*}
    X_\epsilon = \Sigma_1^{-\frac{1}{2}} \left( \Sigma_1^{\frac{1}{2}} \Sigma_2 \Sigma_1^{\frac{1}{2}} + \frac{\epsilon^2}{16} I \right)^{\frac{1}{2}} \Sigma_1^{-\frac{1}{2}} - \frac{\epsilon}{4} \Sigma_1^{-1}\,.
\end{equation*}
%$X_\epsilon = \Sigma_1^{-\frac{1}{2}} \left( \Sigma_1^{\frac{1}{2}} \Sigma_2 \Sigma_1^{\frac{1}{2}} + \frac{\epsilon^2}{16} I \right)^{\frac{1}{2}} \Sigma_1^{-\frac{1}{2}} - \frac{\epsilon}{4} \Sigma_1^{-1}$.
Moreover, it holds that $W_{\epsilon}(\mathbb{P}_1, \mathbb{P}_2) = G_\epsilon(\Sigma_1, \Sigma_2)$, that is, the Sinkhorn divergence coincides with the entropy-regularized Gelbrich divergence.
% Given $\Prob_1\sim\mathcal{N}(0, \Sigma_1)$ and $\Prob_2
% \sim\mathcal{N}(0, \Sigma_2)$ the optimal transport plan for \eqref{eq:sinkhorn} is Gaussian given by $\gamma_0 \sim\mathcal{N}\left(0, \begin{bmatrix} \Sigma_1 & \Sigma_1X_\epsilon\\ X_\epsilon\Sigma_1 &\Sigma_2
% \end{bmatrix}\right)$, where $X_\epsilon = \Sigma_1^{-\frac{1}{2}}\left(\Sigma_1^{\frac{1}{2}}\Sigma_2\Sigma_1^{\frac{1}{2}} + \frac{\epsilon^2}{16}I\right)^
% {\frac{1}{2}}\Sigma_1^{-\frac{1}{2}} - \frac{\epsilon}{4}\Sigma_1^{-1}$. Moreover, the entropy-regularized Gelbrich and Sinkhorn discrepancies coincide, i.e. $W_{\epsilon}(\Prob, \Q) = G_\epsilon(\Sigma_1, \Sigma_2)$.
\label{proposition:tightness}
\end{proposition}
\begin{proof}
% From \cref{remark:equivalence} we know that the optimal transport plan for \eqref{eq:sinkhorn} is the same as the one of \cite[Theorem 2]{del2020statistical}. 
As observed in \cref{remark:equivalence}, regularizing \eqref{eq:sinkhorn} with the negative differential entropy or a KL term leads to the same optimal transport plan. Hence, the expression for $\gamma_0$ follows from \cite[Theorem 2.2]{del2020statistical}. Substituting this optimal coupling in the definition of Sinkhorn discrepancy in \eqref{eq:sinkhorn}, we obtain 
\begin{align*}
    \E_{\gamma_0}[\|x-y\|^2] &= \Tr(\Sigma_1) + \Tr(\Sigma_2) - 2\Tr(\Sigma_1X_\epsilon)
    \\
    \mathrm{KL}(\gamma_0 \| \Prob_1\times\nu) &= \frac{1}{2}\left[\Tr(\Sigma^{-1}\Sigma_2) - d + \log\frac{|\Sigma|}{|\Sigma_2|} + \right.
    \\
    &\left.+\log\left(\left(\frac{2}{\epsilon}\right)^d\left|D_\epsilon + \frac{\epsilon}{4}I\right|\right)\right].
\end{align*}
By inspection, combining the expressions above as per \eqref{eq:sinkhorn} yields \eqref{eq:entropy_reg_gelbrich}, which concludes the proof.
\end{proof}
Let us define the matrices $\mathbf{M} = \mathbf{R}+\mathbf{H}^{\top} \mathbf{Q H} \in \mathbb{R}^{m T}$, $\mathbf{F}_1 = \mathbf{D}^\top\mathbf{U}^\top\mathbf{R}\mathbf{U}\mathbf{D} +(\mathbf{HUD}+\mathbf{G})^\top\!\mathbf{Q}(\mathbf{HUD} + \mathbf{G}) \in \mathbb{S}_+^{d(T+1)}$, and $\mathbf{F}_{2}=\mathbf{U}^{\top} \mathbf{R} \mathbf{U}+\mathbf{U}^{\top} \mathbf{H}^{\top} \mathbf{Q H U} \in \mathbb{S}_+^{p T}$ for brevity. Moreover, we denote the set of causal feedback matrices by $\mathcal{U}^\text{lin}$. With this notation in place, we are now ready to reformulate \eqref{eq:lower_bound} as a finite-dimensional optimization problem. 

\begin{proposition}
\label{proposition:trace_rewriting}
The lower bound \eqref{eq:lower_bound} to the dual problem \eqref{eq:dual} is equivalent to the finite-dimensional optimization problem
\begin{equation}
\label{eq:trace_lowerbound}
    \underline{d}^\star = \max _{\substack{\mathbf{W}\in\mathcal{G}_W\\ \mathbf{V}\in\mathcal{G}_V}}\ \min _{\substack{\mathbf{U}\in\mathcal{U}^{\text{lin}}\\ \mathbf{q}\in\R^{mT}}}\ \Tr\left(\mathbf{F}_{1} \mathbf{W}+\mathbf{F}_{2} \mathbf{V}\right)+\mathbf{q}^{\top} \mathbf{M} \mathbf{q},
    \end{equation}
    where $\mathbf{F}_1$ and $\mathbf{F}_2$ depend on $\bf U$, and the finite-dimensional sets $\mathcal{G}_W$ and $\mathcal{G}_V$ are defined as
    \begin{align}
        \mathcal{G}_{W} &= \big\{ \mathbf{W} \in \Symm_{++}^{d(T+1)} :\ 
        \mathbf{W} = \operatorname{diag}(X_0, W_0, \ldots, W_{T{-}1}),\nonumber 
        \\
        &\quad\quad X_0 \in \mathbb{S}_{++}^d,  W_t \in \Symm_{++}^d,\
        G_\epsilon(\hat{X}_0, X_0) \leq \rho_{x_0},\nonumber \\
        &\quad\quad G_\epsilon(\hat{W}_t, W_t) \leq \rho_{w_t} \ \forall t \in [T] \big\},\label{eq:GW} \\
        \mathcal{G}_{V} &= \big\{ \mathbf{V} \in \Symm_{++}^{pT} :\ 
        \mathbf{V} = \operatorname{diag}(V_0, \ldots, V_{T{-}1}),\
        V_t \in \mathbb{S}_{++}^p\nonumber, \\
        &\quad\quad G_\epsilon(\hat{V}_t, V_t) \leq \rho_{v_t} \ \forall t \in [T] \big\}.\label{eq:GV}
    \end{align}
\end{proposition}

\begin{proof}
We first observe that, for any fixed $\Prob\in\mathcal{S_N}$, the inner minimization in \eqref{eq:lower_bound} constitutes a standard LQG problem, for which linear policies are globally optimal \cite{bertsekas2012dynamic}. Hence, as discussed in \cref{sec:purified-output}, we restrict the inner minimization in \eqref{eq:lower_bound} to policies of the form $\bf u = U\bm\eta + q$, where $\mathbf{q} \in \R^{mT}$ and $\bf U \in \mathcal{U}^\text{lin}$, without loss of generality.

Then, we note that, by \cref{proposition:tightness}, the set $\mathcal{S_N}$ is equivalent to the entropy-regularized Gelbrich set
\begin{equation}
\label{eq:Gelbrich}
\mathcal{G} = \mathcal{G}_{x_{0}} \otimes \Big(\otimes_{t=0}^{T-1} \mathcal{G}_{w_{t}}\Big) \otimes \Big(\otimes_{t=0}^{T-1} \mathcal{G}_{v_{t}}\Big)\,,
\end{equation}
where each component $\mathcal{G}_{x_{0}}$, $\mathcal{G}_{w_{t}}$, and $\mathcal{G}_{v_{t}}$ is defined by
    \begin{align*}
        \mathcal{G}_{x_0}\! &=\! \big\{ \mathbb{P}_{x_0}\! \in\! \mathcal{P}_0(\mathbb{R}^d) \!: \!
        \mathbb{E}_{\mathbb{P}}[x_0 x_0^\top] = X_0,G_\epsilon(\hat{X}_0, X_0)\! \leq \!\rho_{x_0} \big\}, \nonumber\\
        \mathcal{G}_{w_t}\! &=\!\big\{ \mathbb{P}_{w_t}\! \in\! \mathcal{P}_0(\mathbb{R}^d) \!:\!\mathbb{E}_{\mathbb{P}}[w_t w_t^\top] = W_t, G_\epsilon(\hat{W}_t, W_t)\! \leq\! \rho_{w_t} \big\}, \nonumber\\
        \mathcal{G}_{v_t}\! &=\!\big\{ \mathbb{P}_{v_t}\! \in\! \mathcal{P}_0(\mathbb{R}^p) \!:\!\mathbb{E}_{\mathbb{P}}[v_t v_t^\top] = V_t, G_\epsilon(\hat{V}_t, V_t)\! \leq\! \rho_{v_t} \big\}.
    \end{align*}
    %where $\hat{X}_0$, $\hat{W}_t$, and $\hat{V}_t$ are the covariance matrices of the exogenous disturbances $x_0$, $w_t$, and $v_t$ under the nominal distribution $\hat{\Prob}$, respectively. 
    Combining these observations, we equivalently rewrite \eqref{eq:lower_bound} as
\begin{align*}
&\max_{\mathbb{P} \in \mathcal{G}}\ \min_{\bf U, q}\ \mathbb{E}_{\mathbb{P}}\left[\bf u^{\top} R u + x^{\top} Q x \right]
\\
&\text { s.t. }\,\, \mathbf{U} \in \mathcal{U}^\text{lin},\ \bf {u = U(Dw + v) + q}, ~~\bf {x = H u + G w}.
\end{align*}
Following the same argument as \cite[Proposition 3.2]{taskesen2023distributionally}, the proof is concluded by rewriting the expectation of a quadratic form as a trace and by replacing the ambiguity set $\mathcal{G}$ by \eqref{eq:GW} and \eqref{eq:GV}.
%Furthermore, we can express the expectation of the quadratic form as a trace. We conclude the proof observing that, without loss of generality, the ambiguity set $\mathcal{G}$ can be replaced by \eqref{eq:GW} and \eqref{eq:GV}. 
% Indeed, while $\mathcal{G}$ defines a larger feasible set, we note that the probability
% $\Prob = \Prob_{x_0} \times \left( \prod_{t=0}^{T-1} \Prob_{w_t} \right) \times \left( \prod_{t=0}^{T-1} \Prob_{v_t} \right)$
% defined through \( \Prob_{x_0} = \mathcal{N}(0, X_0) \), \( \Prob_{w_t} = \mathcal{N}(0, W_t) \), and \( \Prob_{v_t} = \mathcal{N}(0, V_t) \), for \( t \in [T] \) is a feasible solution with the same objective value.
%{\color{red}The proof is concluded by expressing the expectation of the quadratic form as a trace and observing that, without loss of generality, the ambiguity set $\mathcal{G}$ can be taken as in \eqref{eq:trace_lowerbound}, following the same argument as \cite[Proposition 3.2]{taskesen2023distributionally}.}
\end{proof}

\subsection{Construction of an upper bound for \texorpdfstring{$p^\star$}{p*}}
% We construct an upper bound for $p^\star$ by suitably enlarging the ambiguity set $\mathcal{S}$ and restricting the input to linear policies. We enlarge 
% the ambiguity set forgetting all information about the distributions in $\mathcal{S}$ except for their covariances. To this end we will use the entropy-regularized Gelbrich discrepancy and the following Proposition.
To derive an upper bound for $p^\star$, we restrict our attention to linear policies in $\mathcal{U}_\eta$ and appropriately enlarge the ambiguity set $\mathcal{S}$. This construction relies on the following result, which provides a lower bound to the Sinkhorn discrepancy between two distributions when one of them is Gaussian.
%which draws a connection between Sinkhorn and entropy-regularized Gelbrich discrepancies.
\begin{proposition}
    % If $\Prob$ is a zero-mean Gaussian distribution with covariance matrix $\Sigma_1$ and $\Q$ is a zero-mean probability on $\R^d$ with covariance given by $\Sigma_2$
    % then $W_{\epsilon}(\Prob, \Q) \geq G_\epsilon(\Sigma_1, \Sigma_2)$.
    Let $\Prob\sim \mathcal{N}(0, \Sigma_1)$, $\Sigma_1 \in \Symm_{++}^d$, and $\Q \in \mathcal{P}_0(\R^d)$ be a distribution with covariance $\Sigma_2 \in \Symm_{++}^d$. Then, it holds that $W_{\epsilon}(\Prob, \Q) \geq G_\epsilon(\Sigma_1, \Sigma_2)$.
    \label{proposition:Gelbrich}
\end{proposition}
\begin{proof}
Recall the definition of the joint probability distribution $\gamma_0$ given in \cref{proposition:tightness}. For any  $\gamma\in\Gamma(\Prob, \Q)$, the objective in \eqref{eq:sinkhorn} can be rewritten as
\begin{align*}
    &\epsilon\int_{\R^d\times\R^d}\log\left(\frac{\dd\gamma(x,y)e^{\frac{\|x-y\|^2}{\epsilon}}}{\dd\Prob(x)\dd\nu(y)}\right)\dd\gamma(x,y)
    \\
    &=\epsilon \mathrm{KL}(\gamma\|\gamma_0) + \epsilon\int_{\R^d\times\R^d}\log\left(\frac{\dd\gamma_0(x,y)e^{\frac{\|x-y\|^2}{\epsilon}}}{\dd\Prob(x)\dd\nu(y)}\right)\dd\gamma(x,y)
    \\
    &=\epsilon \mathrm{KL}(\gamma\|\gamma_0) + \epsilon\int_{\R^d\times\R^d}\log\left(\dd\gamma_0(x,y)e^{\frac{\|x-y\|^2}{\epsilon}}\right)\dd\gamma(x,y) 
    \\
    & \quad 
    - \epsilon\int_{\R^d}\log\left(\dd\Prob(x)\right)\dd\Prob(x) 
    - \epsilon\int_{\R^d}\log\left(\dd\nu(y)\right)\dd\Q(y)
    \\
    &=\underbrace{\Tr(\Sigma_1 + \Sigma_2)\! -\! 2\Tr(\Sigma_1X_\epsilon)
    \!-\! \tfrac{\epsilon}{2}\log\!\left((2\pi e)^{2d}\!\left(\tfrac{\epsilon}{2}\right)^d\!|\Sigma_1X_\epsilon|\right)}_{(\spadesuit)}
    \\
    &
    +\! \underbrace{\tfrac{\epsilon}{2}\log\!\left((2\pi e)^d|\Sigma_1|\right)}_{(\clubsuit)}
    \!+\! \underbrace{\tfrac{\epsilon}{2}\left(\Tr(\Sigma^{-1}\Sigma_2)\! -\! d
    \!+\! \log\!\left((2\pi e)^d|\Sigma|\right)\right)}_{(\diamondsuit)}
    \\
    &
    + \epsilon \mathrm{KL}(\gamma\|\gamma_0),
\end{align*}
where $(\spadesuit)$ is obtained computing the integral with the explicit expression of the density of $\gamma_0$ as per \cref{proposition:tightness}, $(\clubsuit)$ is the differential
entropy of the Gaussian $\Prob$, and $(\diamondsuit)$ results by computing the cross-entropy between $\nu$ and $\Q$. After some algebraic manipulations, we obtain $\Tr(\Sigma_1X_\epsilon) = \Tr(D_\epsilon) - \frac{\epsilon d}{4}$. Moreover, we also have that \[-\log\left(\left(\tfrac{\epsilon}{2}\right)^d|X_\epsilon|\right) = \log\left(\left(\tfrac{2}{\epsilon}\right)^d|X_\epsilon^{-1}\Sigma_2|\right) - \log|\Sigma_2|,\] and, using the relationships in \cite[Proposition 2.1]{del2020statistical}, that
\begin{align*}
  \log|X_\epsilon^{-1}\Sigma_2| &= \log|X_\epsilon\Sigma_1 + \tfrac{\epsilon}{2}I| 
  \\
  &=\log|\Sigma_1^{-\frac{1}{2}}D_\epsilon\Sigma_1^{\frac{1}{2}} + \tfrac{\epsilon}{4}I| = \log|D_\epsilon+ \tfrac{\epsilon}{4}I|.
\end{align*}
Therefore, we conclude that
\begin{align*}
&(\spadesuit) + (\clubsuit) + (\diamondsuit) =\ 
\Tr(\Sigma_1) + \Tr(\Sigma_2) - 2\Tr(D_\epsilon)+
\\
&\frac{\epsilon}{2}\Tr(\Sigma^{-1}\Sigma_2) + \frac{\epsilon}{2}\log\frac{|\Sigma|}{|\Sigma_2|} + \frac{\epsilon}{2}\log\!\left(\left(\frac{2}{\epsilon}\right)^d
\left|D_\epsilon + \frac{\epsilon}{4}I\right|\right).
    % (\spadesuit) + (\clubsuit) + (\diamondsuit) =  \Tr(\Sigma_1) + \Tr(\Sigma_2) -2 \Tr(D_\epsilon)+\frac{\epsilon}{2}\Tr\left(\Sigma^{-1}\Sigma_2\right) + 
    % \frac{\epsilon}{2}\log\frac{|\Sigma|}{|\Sigma_2|} +\frac{\epsilon}{2}\log\left(\left(\frac{2}{\epsilon}\right)^d\left|D_\epsilon + \frac{\epsilon}{4}I\right|\right).
\end{align*}
% Then, the claim of the proposition follows since $\mathrm{KL}(\gamma\|\gamma_0)\geq 0$ and because $\gamma$ was arbitrary.
The claim follows from the nonnegativity of $\mathrm{KL}(\gamma\|\gamma_0)$ and the arbitrariness of $\gamma$.
\end{proof}

\cref{proposition:Gelbrich} shows that the Sinkhorn discrepancy can be lower bounded by discarding all distributional information except for the covariances. Hence, a valid outer approximation for the set $\mathcal{S}$ is given by the entropy-regularized Gelbrich set \eqref{eq:Gelbrich}, as formalized in the following Corollary.% which is a direct consequence of \cref{proposition:Gelbrich}.
% To construct the upper bound, consider the ambiguity set in \eqref{eq:Gelbrich}. This is a valid outer approximation for the set $\mathcal{S}$ as proved in the following Corollary which is a direct consequence of \cref{proposition:Gelbrich}.
\begin{corollary}
    For any regularization parameter $\epsilon\geq0$ and radius $\rho \geq 0$, it holds that the Sinkhorn ambiguity set $\{\Q \in \mathcal{P}_0(\mathbb{R}^d): W_\epsilon(\Prob, \Q) \leq \rho\}$ is always contained in the entropy-regularized Gelbrich set $\{\Q \in \mathcal{P}_0(\mathbb{R}^d): \E_\Q[zz^\top] = \Sigma_2\in\Symm^d_{++},\  G_\epsilon(\Sigma_1, \Sigma_2) \leq \rho\}$.
\end{corollary}
\begin{proof}
By \cref{proposition:Gelbrich}, we have that $W_{\epsilon}(\Prob,\Q) \geq G_\epsilon(\Sigma_1,\Sigma_2)$. Hence, if $W_{\epsilon}(\Prob,\Q) \leq \rho$, then $G_\epsilon(\Sigma_1,\Sigma_2) \leq \rho$, which completes the proof.
% Since $W_{\epsilon}(\Prob, \Q) \geq G_\epsilon(\Sigma_1, \Sigma_2)$ by \cref{proposition:Gelbrich}, for any distribution $\Q$ such that $W_{\epsilon}
% (\Prob, \Q) \leq \rho$ it follows that $G_\epsilon(\Sigma_1, \Sigma_2) \leq \rho$ which concludes the proof. 
\end{proof}
By restricting to linear policies, we finally construct an upper bound for \eqref{eq:primal} as
% To arrive at an upper bound to problem \eqref{eq:primal} with finitely many optimization variables, {\color{cyan} we finally}  We define the upper bound problem as
\begin{equation}
    \label{eq:upper_bound}
    \overline{p}^\star =
    \left\{
    \begin{aligned}
    &\min_{\bf U, q}\ \max_{\mathbb{P} \in \mathcal{G}}\ \mathbb{E}_{\mathbb{P}}\left[\bf u^{\top} R u + x^{\top} Q x \right]
    \\
    &\text { s.t. }\ \mathbf{U} \in \mathcal{U}^\text{lin},\ \bf {u = U(Dw + v) + q},
    \\
    &\quad\quad\,\, \bf {x = H u + G w}.
    \end{aligned}
    \right.
\end{equation} 
% \begin{equation}
%     \label{eq:upper_bound}
%     \overline{p}^\star =
%     \left\{
%     \begin{aligned}
%     &\min_{\mathbf{U}, q}\ \max_{\mathbb{P} \in \mathcal{G}}\ \mathbb{E}_{\mathbb{P}}\left[\mathbf{u}^{\top} R \mathbf{u} + \mathbf{x}^{\top} Q \mathbf{x} \right]
%     \\
%     &\text{s.t. } \mathbf{U} \in \mathcal{U}_{\eta}^{\text{lin}},\quad \mathbf{u} = \mathbf{U}(D\mathbf{w} + \mathbf{v}) + q,
%     \\
%     &\quad\ \mathbf{x} = H\mathbf{u} + G\mathbf{w}
%     \end{aligned}
%     \right.
% \end{equation}
As we enlarged nature's subproblem feasible set, and at the same time shrank the possible control
laws that the designer can select, we have that $\overline{p}^\star \geq p^\star$. In the next proposition, we rewrite \eqref{eq:upper_bound} as a finite-dimensional optimization problem.
\begin{proposition}
    The upper bound \eqref{eq:upper_bound} to the primal problem \eqref{eq:DR-LQG} is equivalent to the finite-dimensional program
    \begin{equation}
    \label{eq:trace_upperbound}
    \overline{p}^\star = \min _{\substack{\mathbf{U}\in\mathcal{U}^{\text{lin}}\\ \mathbf{q}\in\R^{mT}}}\ \max _{\substack{\mathbf{W}\in\mathcal{G}_W\\ \mathbf{V}\in\mathcal{G}_V}}\ \Tr\left(\mathbf{F}_{1} \mathbf{W}+\mathbf{F}_{2} \mathbf{V}\right)+\mathbf{q}^{\top} \mathbf{M} \mathbf{q}
    \end{equation}
where $\mathbf{F}_1, \mathbf{F}_2, \mathcal{G}_W, \mathcal{G}_V$ are defined as in \cref{proposition:trace_rewriting}.
\end{proposition}
\begin{proof}
The proof follows the same steps as the proof of \cite[Proposition 3.4]{taskesen2023distributionally}. In the same way as in \cref{proposition:trace_rewriting}, we first substitute $\bf {x = H u + G w}$ and $\bf {u = U(Dw + v) + q}$ in the objective of \eqref{eq:upper_bound}. Then, for any $\Prob\in\mathcal{G}$, we rewrite the previous expectation as a trace in terms of the covariance matrices $\mathbf{W} = \E_\Prob[\bf ww^\top]$ and $\mathbf{V} = \E_\Prob[\bf vv^\top]$, and conclude by replacing $\mathcal{G}$ with \eqref{eq:GW} and \eqref{eq:GV}.
\end{proof}

We conclude this section by observing that \eqref{eq:trace_lowerbound} and \eqref{eq:trace_upperbound} are dual to each other, in the sense that one can be obtained from the other by swapping the order of optimization.

\subsection{Existence of a Nash equilibrium}
In this section, we show that strong duality holds between \eqref{eq:trace_lowerbound} and \eqref{eq:trace_upperbound}, and hence that the Sinkhorn DR LQG admits a Nash equilibrium in the class of linear policies. Before proving these results, we present a technical lemma that characterizes structural and
topological properties of $\mathcal{G}_W$ and $\mathcal{G}_V$.
\begin{lemma}
    \label{lemma:compactness}
    Given $\hat{\Sigma}\in\Symm^d_{++}$, $\rho\geq 0$ and $\epsilon \geq 0$ and finite, the set $\mathcal{D} = \{M \in\Symm^d_{++} : G_\epsilon(\hat{\Sigma},M) \leq \rho\}$ is convex and compact.
\end{lemma}
\begin{proof}
    \textit{Convexity}: the map $M\rightarrow G_\epsilon(\hat{\Sigma}, M)$ is convex because sum of convex functions. Indeed, 
    % \begin{equation*}
    %     G_\epsilon(\hat{\Sigma}, M) = \underbrace{\Tr(\hat{\Sigma}) + \Tr(M) -2 \Tr(D_\epsilon) + \frac{\epsilon}{2}\log\left(\left(\frac{2}{\epsilon}\right)^d \left|D_\epsilon + \frac{\epsilon}{4}I\right|\right)}_{(\heartsuit)}
    %     + \frac{\epsilon}{2}\Tr\left(\Sigma^{-1}M\right) + \frac{\epsilon}{2}\log\frac{|\Sigma|}{|M|},
    % \end{equation*}
    % \begin{equation*}
    % \begin{aligned}
    % G_\epsilon(\hat{\Sigma}, M)\! =\ 
    % &\!\underbrace{\Tr(\hat{\Sigma} + M- 2D_\epsilon)\! +\! \tfrac{\epsilon}{2}\log\!\left(\!\left(\tfrac{2}{\epsilon}\right)^d\!\left|D_\epsilon\! +\! \tfrac{\epsilon}{4}I\right|\!\right)}_{(\heartsuit)} \\
    % & + \frac{\epsilon}{2}\Tr\!\left(\Sigma^{-1}M\right) + \frac{\epsilon}{2}\log\frac{|\Sigma|}{|M|},
    % \end{aligned}
    % \end{equation*}
     \[
    \scalebox{0.75}{$
    G_\epsilon(\hat{\Sigma}, M)\! = 
    \!\underbrace{\Tr(\hat{\Sigma}\! +\! M\!-\! 2D_\epsilon)\! +\! \tfrac{\epsilon}{2}\left(\log\!\left(\!\left(\tfrac{2}{\epsilon}\right)^d\!\left|D_\epsilon\! +\! \tfrac{\epsilon}{4}I\right|\!\right)\right.}_{(\heartsuit)}\!
    \left.\!+\Tr\!\left(\Sigma^{-1}M\right)\! +\! \log\tfrac{|\Sigma|}{|M|}\right)\!,$}
    \]
    and $(\heartsuit)$ is convex because of \cite[Proposition 6]{janati2020entropic}, while the trace is linear (hence convex) and the negative log-determinant is convex. This implies that $\mathcal{D}$ is convex because it is the level set of a convex function \cite[Proposition 2.7]{rockafellar2009variational}.
    %Then $\mathcal{D}$ is
    % convex because it is the level set of a convex function \cite[Proposition 2.7]{rockafellar2009variational}.
    \newline
    \textit{Compactness}: we want to show that $\mathcal{D}$ is closed and bounded. To this end, let the function $f:\Symm^d_{++}\rightarrow\R$ 
    be defined as \[
\scalebox{0.88}{$f(M) = \Tr(M - 2D_\epsilon) + \frac{\epsilon}{2} \left( \log|D_\epsilon + \frac{\epsilon}{4}I| + \Tr(\Sigma^{-1}M) - \log|M| \right)$.}
\]
    %$f(M) = \Tr(M) -2 \Tr(D_\epsilon) + \frac{\epsilon}{2}\left(\log|D_\epsilon + \frac{\epsilon}{4}I| + \Tr(\Sigma^{-1}M) - \log|M|\right)$.
    This function involves affine transformations of $M$ along with continuous
    transformations on $\Symm^d_{++}$ such as matrix square-root, trace and log-determinant. Hence, $f(\cdot)$ is continuous and the set $\mathcal{D} = \{M \in \Symm^d_{++} : f(M) \leq \tilde{\rho}\}$ with $\tilde{\rho} = \rho - \Tr(\hat{\Sigma}) + 
    \frac{\epsilon}{2}(\log(\frac{\epsilon}{2})^d - \log|\Sigma|) $ is closed because it is the lower level set of a continuous function \cite[Theorem 1.6]{rockafellar2009variational}.

    To show boundedness, we proceed by contradiction and assume that $\sup_{M\in\mathcal{D}}\ \lambda_{\max}(M) = +\infty$. We construct a lower bound for $f(M)$ by bounding each addend separately.
    % \begin{equation*}
    %     \underbrace{\Tr(M)}_{(\rm i)}\! -\!2 \underbrace{\Tr(D_\epsilon)}_{(\rm ii)} + \frac{\epsilon}{2}\Bigl(\underbrace{\log\left|D_\epsilon\! +\! \frac{\epsilon}{4}I\right|}_{(\rm iii)}
    %     \!+\! \underbrace{\Tr(\Sigma^{-1}M)}_{(\rm iv)}\! -\! \underbrace{\log|M|}_{(\rm v)}\Bigr).
    % \end{equation*}
   %  \begin{equation*}
   % f(M) =  \underbrace{\Tr(M)}_{(\mathrm{i})} \!-\! 2\underbrace{\Tr(D_\epsilon)}_{(\mathrm{ii})} 
   %  + \tfrac{\epsilon}{2}\bigl(
   %  \underbrace{\log\!\left|D_\epsilon \!+\! \tfrac{\epsilon}{4}I\right|}_{(\mathrm{iii})}
   %  \!+\! \underbrace{\Tr(\Sigma^{-1}M)}_{(\mathrm{iv})}
   %  \!-\! \underbrace{\log|M|}_{(\mathrm{v})}
   %  \bigr).
   %  \end{equation*}
    The term $\Tr(M)$ can be lower bounded by $\lambda_{\max}(M)$. Since $\|\hat{\Sigma}^{\frac{1}{2}}M\hat{\Sigma}^{\frac{1}{2}}\|_2 \leq \|\hat{\Sigma}^{\frac{1}{2}}\|_2^2\|M\|_2 = 
    \lambda_{\max}(\hat{\Sigma})\lambda_{\max}(M)$ by the submultiplicativity property of the operator norm, $\Tr(D_\epsilon) = \sum_{i=1}^d \sqrt{\lambda_i(\hat{\Sigma}^\frac{1}{2}M\hat{\Sigma}^\frac{1}{2}) + \frac{\epsilon^2}{16}}
    \leq d\sqrt{\lambda_{\max}(\hat{\Sigma})\lambda_{\max}(M) + \frac{\epsilon^2}{16}}$, and we obtain a lower bound for the second addend. By definition of $D_\epsilon$ in \cref{def:entropy_reg_gelbrich}, since $\hat{\Sigma}^\frac{1}{2}M\hat{\Sigma}^\frac{1}{2} \succ 0$, we have that $\lambda_i(D_\epsilon) > \frac{\epsilon}{4}\, \forall i$. Therefore, $\left|D_\epsilon + \frac{\epsilon}{4}I\right| \geq \left(\frac{\epsilon}{2}\right)^d$ and we can lower bound the third addend by $\frac{\epsilon d}{2}\log\left(\frac{\epsilon}{2}\right)$. The fourth addend can be bounded as
    $\Tr(\Sigma^{-1}M) \geq \frac{\Tr(M)}{\lambda_{\max}(\Sigma)} \geq \frac{\lambda_{\max}(M)}{\lambda_{\max}(\Sigma)}$. Finally, $\log|M| = \sum_{i=1}^d \log(\lambda_i(M)) \leq d\log(\lambda_{\max}(M))$ and $\log\frac{|\Sigma|}{|M|}$ is lower-bounded by $-d\log|\Sigma|\log\lambda_{\max}(M)$. Putting
    everything together we get
    % \begin{equation*}
    %     f(M) \geq \lambda_{\max}(M) - 2d\sqrt{\lambda_{\max}(\hat{\Sigma})\lambda_{\max}(M) + \frac{\epsilon^2}{16}} + \frac{\epsilon d}{2}\log\left(\frac{\epsilon}{2}\right) +
    %     \frac{\epsilon}{2}\frac{\lambda_{\max}(M)}{\lambda_{\max}(\hat{\Sigma})} - \frac{\epsilon d}{2}\log(\lambda_{\max}(M)).
    % \end{equation*}
    \begin{equation*}
    \begin{aligned}
        &f(M) \geq\ \lambda_{\max}(M) 
        - 2d \sqrt{ \lambda_{\max}(\hat{\Sigma}) \lambda_{\max}(M) + \frac{\epsilon^2}{16} } 
        \\
        &+ \frac{\epsilon d}{2} \log\left( \frac{\epsilon}{2} \right)  + \frac{\epsilon}{2} \frac{\lambda_{\max}(M)}{ \lambda_{\max}(\hat{\Sigma}) } 
        - \frac{\epsilon d}{2}\log|\Sigma| \log\left( \lambda_{\max}(M) \right).
    \end{aligned}
    \end{equation*}
    % It is clear that the linear terms in $\lambda_{\max}(M)$ dominate the other terms as $\lambda_{\max}(M)\rightarrow +\infty$. 
    As $\lambda_{\max}(M)\rightarrow +\infty$, the linear terms in $\lambda_{\max}(M)$ dominate the other ones. Consequently, the RHS is unbounded
    when choosing $M$ such that $\lambda_{\max}(M)= +\infty$. Therefore, since  $\tilde{\rho}$ is finite when so is $\epsilon$, 
    we contradict the fact that any $M\in\mathcal{D}$ satisfies $f(M) \leq \tilde{\rho}$. Hence, $\mathcal{D}$ is bounded.
\end{proof}
\cref{lemma:compactness} is key, as it enables the use of Sion's minimax theorem \cite[Theorem 3.4]{Sion1958} to prove the existence of a Nash equilibrium for \eqref{eq:DR-LQG}. This is formally stated in the next theorem. 
% \cref{lemma:compactness} enables us to prove the existence of a Nash equilibrium, which is formally stated in the following theorem.
\begin{theorem}
\label{thm:Nash}
The following results hold:
\begin{enumerate}[leftmargin=*]
    % \item The optimal values of \eqref{eq:lower_bound} and \eqref{eq:upper_bound} coincide, that is $\overline{p}^\star = \underline{d}^\star$;
    % \item The optimal values of \eqref{eq:primal} and \eqref{eq:dual} coincide, that is $p^\star = d^\star$;
    \item The optimal values $\overline{p}^\star$ of \eqref{eq:lower_bound} and $\underline{d}^\star$ of \eqref{eq:upper_bound} coincide;
    \item The optimal values $p^\star$ of \eqref{eq:primal} and $d^\star$ of \eqref{eq:dual} coincide;
    \item There exist $\mathbf{U}^\star\in\mathcal{U}^\text{lin}$ and $\mathbf{q}^\star\in\R^{mT}$ such that the DR LQG problem \eqref{eq:primal} is solved by $\bf u = q^\star + U ^\star y$;
    \item The dual DR problem \eqref{eq:dual} is solved by a Gaussian distribution $\Prob^\star\in\mathcal{S_N}$. 
\end{enumerate}
\end{theorem}
\begin{proof}
By Lemma \ref{lemma:compactness} the sets $\mathcal{G}_W$ and $\mathcal{G}_V$ are compact and convex. The trace is linear, hence the objective function is
concave in $\mathbf{W}$ and $\mathbf{V}$. Moreover, since $\mathbf{Q}, \mathbf{R}, \mathbf{M} \succeq 0$, it is convex in $\mathbf{U}$ and $\mathbf{q}$. Therefore, 
we can apply Sion's minimax theorem \cite[Theorem 3.4]{Sion1958} to show that strong duality holds, proving the first point of the theorem. 
By strong duality, the chain of inequalities $\underline{d}^\star \leq d^\star \leq p^\star \leq \overline{p}^\star$ collapses to equalities proving the second point.
% The second point
% is a consequence of the first one and the chain of inequalities $\underline{d}^\star \leq d^\star \leq p^\star \leq \overline{p}^\star$ that collapses to equalities.
The equality 
$p^\star = \overline{p}^\star$ implies that \eqref{eq:primal} is solved by a linear causal policy of the purified observations. However, as pointed out in \cref{sec:purified-output}, any causal controller that is linear in the purified outputs $\bm\eta$ can be also expressed as a causal linear feedback in the measurements $\bf y$. This proves the third point. We finish the proof noticing that the last point of the theorem follows from the identity $\underline{d}^\star = d^\star$. 
\end{proof}
\Cref{thm:Nash} implies that \eqref{eq:DR-LQG} is solved by an LQG with worst-case covariance matrices given by \eqref{eq:trace_lowerbound}. We report a reformulation of \eqref{eq:trace_lowerbound} as finite-dimensional conic program in \cref{prop:SDP_reformulation} in the Appendix.
We conclude this section by observing that, when $\epsilon\rightarrow\infty$, the set $\mathcal{S}$ either becomes empty or reduces to the singleton $\nu$ depending on whether $\Tr(\hat{\Sigma}) + \Tr(\Sigma)$ exceeds $\rho$ or not, see \cite[Proposition 1]{cescon2025data}. In particular, when $\Tr(\hat{\Sigma}) + \Tr(\Sigma) \leq \rho$, we still retain global optimality of linear policies as $\mathcal{S} = \{\nu\}$, and $\nu$ is Gaussian. Instead, when $\Tr(\hat{\Sigma}) + \Tr(\Sigma) > \rho$, the Sinkhorn DR LQG problem \eqref{eq:DR-LQG} becomes infeasible.
% {\color{red}Hence, if $\int\|x-y\|^2\dd\hat{\Prob}(x)\dd\nu(y) = \Tr(\hat{\Sigma}) + \Tr(\Sigma) \leq \rho$ with $\hat{\Sigma}$ the nominal covariance, we still retain that the linear policies are globally optimal because the set $\mathcal{S}$ shrinks to the Gaussian reference $\nu$.
% HAI DESCRITTO DUE CASI, POI NE COMMENTI UNO, SEMBRA CHE NASCONDI COSE SOTTO IL TAPPETO. DIREI IN BOTH CASES WE RETAIN OPTIMALITY; WHEN CONDITION GREATER THAN RHO IT IS TRIVIAL BECAUSE INFEASIBLE, WHEN CONDITION LESS OR EQUAL RHO QUELLO CHE HAI GIA SCRITTO.
% }
%\begin{remark}
    %Notice that from \cite[Proposition 1]{cescon2025data}, when $\epsilon\rightarrow\infty$ the set $\mathcal{S}$ can only be empty or the singleton $\nu$. Hence, if $\int\|x-y\|^2\dd\hat{\Prob}(x)\dd\nu(y) = \Tr(\hat{\Sigma}) + \Tr(\Sigma) \leq \rho$ with $\hat{\Sigma}$ the nominal covariance, we still retain that the linear policies are globally optimal because the set $\mathcal{S}$ shrinks to the Gaussian reference $\nu$.
%\end{remark}
\section{Numerical Experiments}
In this section, we present numerical simulations to showcase the advantages of robustifying against distributional uncertainty.\footnote{All our experiments were run on an M3 Pro CPU machine with 36GB RAM. All SDP problems were modeled in Matlab 2023a using Yalmip and solved with MOSEK.
Our source code is publicly available at \href{https://github.com/DecodEPFL/Optimality_Sinkhorn}{\texttt{https://github.com/DecodEPFL/Optimality\_Sinkhorn.git}}.} For our experiments, we consider the open-loop unstable discrete-time linear dynamical system given by
\begin{align*}
    x_{t+1} = \begin{bmatrix}
        1.1 & 0.1 \\
        0 & 1.1
    \end{bmatrix}x_t + 
    \begin{bmatrix}
        1\\
        1
    \end{bmatrix}
    u_t + w_t, \quad
    y_t = 
    % \begin{bmatrix}
    %     1 && 0\\
    %     0 && 1
    % \end{bmatrix}
    x_t + v_t\,,
\end{align*}
with cost matrices $Q_t = I_2$ and $R_t = 1$ at all times, control horizon $T =25$, and nominal covariances $\hat{X}_0 = I_2, \hat{W}_t = I_2, \hat{V}_t = 0.01I_2, \ \forall t\in[T]$. The reference covariance was picked $\Sigma = I_2$.

In \cref{fig:plots}, we benchmark the classical LQG controller designed based on the nominal covariances against the Sinkhorn DR LQG policy obtained by solving \eqref{eq:DR-LQG} using the conic program in \cref{prop:SDP_reformulation} with radii $\rho_{x_0} = 20,\rho_{w_t} = 0.2, \rho_{v_t} = 20$ and regularization parameter $\epsilon = 0.01$. Specifically, we compare the performance of these controllers on $5000$ realizations of the exogenous disturbances drawn from the nominal Gaussian distribution (on the left) and from the respective nature's adversarial choice of distribution $\Prob$ within $\mathcal{S}$ (on the right). As expected, we observe that the Sinkhorn DR LQG policy incurs a slightly higher average cost when the true distribution corresponds to the nominal one. Conversely, when the noise distributions are selected adversarially within $\mathcal{S}$, we observe that the proposed Sinkhorn DR LQG policy achieves a lower average cost. These results validate our design and highlight fundamental tradeoff between nominal performance and distributional robustness.

%As expected the DR LQG is more conservative and shows a higher cost on average compared to the nominal LQG. here we see that the red bins are more shifted to lower cost values. In the second experiment, instead, we compared the performances of the two controllers against their respective nature's optimal choice of $\Prob$. We extracted 5000 realizations of $\delta$ from this adversary distributions and, as expected, the DR LQG proves to be robust having a smaller cost on average than the nominal LQG. This can be seen in the right plot of \cref{fig:plots} where the red histogram shifts towards bigger cost values compared to the green one.

% \begin{figure}[th!]
%    \centering
%    % \includegraphics[trim=0 0 0 0, clip, width=0.99\linewidth]
%    \includegraphics[trim=0 0 0 0, clip, width=\linewidth]{figures/on_policy.pdf}
%     \caption{}
%    \label{fig:on_policy}
% \end{figure}
% \begin{figure}[h!]
%    \centering
%    % \includegraphics[trim=0 0 0 0, clip, width=0.99\linewidth]
%    \includegraphics[trim=0 0 0 0, clip, width=\linewidth]{figures/worst_policy.pdf}
%     \caption{}
%    \label{fig:worst_policy}
% \end{figure}

\begin{figure*}[htb]
    \centering
    {\includegraphics[width = \columnwidth]{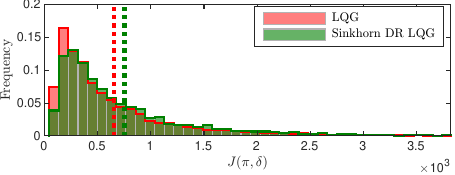}}
    \hfill
    {\includegraphics[width = \columnwidth]{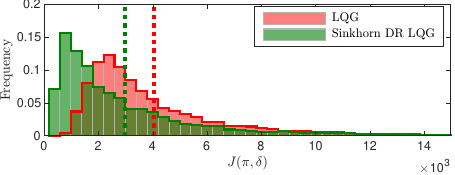}}
  \caption{Comparison between the control cost incurred by the nominal LQG controller (red histograms, in the background) and the proposed Sinkhorn DR LQG policy (green histograms, in the foreground) over $5000$ disturbance realizations drawn from the nominal distribution (on the left) and the respective worst-case distribution in the Sinkhorn ambiguity set (on the right). Dotted vertical lines represent theoretical mean values.}
  \label{fig:plots} 
\end{figure*}
\section{Conclusion}
% In light of recent developments in distributionally robust variants of the classical LQG problem, this work proposed an output-feedback controller for discrete-time stochastic linear systems that minimizes a worst-case cost over Sinkhorn ambiguity sets. Consistent with prior results, our analysis demonstrated that linear policies remain globally optimal, despite the added complexity introduced by the entropy-regularized formulation. Furthermore, the worst-case distribution within the ambiguity set retains a Gaussian form, highlighting the natural compatibility of this robust framework with classical LQG structure.
% This work proposed an output-feedback controller for discrete-time linear systems affected by stochastic noise that minimizes a worst-case cost over Sinkhorn ambiguity sets. Our analysis demonstrated that linear policies remain globally optimal, despite the added complexity introduced by the entropy-regularized formulation. Furthermore, the worst-case distribution within the ambiguity set retains a Gaussian form, highlighting the natural compatibility of this robust framework with classical LQG structure.
In this work, we studied a DR generalization of classical LQG control where the noise distributions are unknown and belong to entropy-regularized Wasserstein or Sinkhorn ambiguity sets. We proved that, despite robustifying the objective in the probability space, nature's distributions retain a Gaussian form and hence linear policies remain globally optimal. We validated the effectiveness of our Sinkhorn DR LQG policy through numerical simulations showing improved robustness compared to classical LQG design. 
\bibliographystyle{IEEEtran}
\bibliography{reference}
\appendix
% \section{Definitions of system matrices}
\section{Appendix}
\label{appendix}
We report below expressions for the stacked cost matrices $\mathbf{Q} \in\Symm_+^{d(T+1)}$ and $\mathbf{R}\in\Symm_{++}^{mT}$, and for the matrices $\mathbf{C} \in \R^{pT\times d(T +1)}$, $\mathbf{G} \in \R ^{d(T +1)\times d(T +1)}$, and $\mathbf{H} \in \R^{d(T +1)\times mT}$ encoding the system linear dynamics and measurement equations. Specifically, we have $\mathbf{Q} = \operatorname{blkdiag} (Q_0, \dots, Q_T)$, $\mathbf{R} = \operatorname{blkdiag} (R_0, \dots, R_{T-1})$ and 
\begin{equation*}
{\small \mathbf{C}=\setlength\arraycolsep{2pt}\begin{bmatrix}
C_0 & 0 &  &  & \\
 & C_1 & \ddots & & \\
 &  & \ddots & \ddots & \\
 &  & & C_{T-1} & 0
\end{bmatrix}\,, \quad
\mathbf{G}=\begin{bmatrix}
    A^0_0 & & & \\
A_0^1 & A_1^1 & & \\
\vdots & & \ddots & \\
A^{T}_0 & A^{T}_1 & \dots & A^T_T
\end{bmatrix}\,,}
\end{equation*}
\begin{equation*}
{\small
\mathbf{H}=
\begin{bmatrix}
0 & & & & \\
A^1_1B_0 & 0 & & & \\
A_1^2B_0 & A_2^2B_1 & 0 & & \\
\vdots & & & \ddots & \\
\vdots & & & & 0 \\
A^{T}_1 B_0 & A^{T}_2 B_1 & \dots & \dots & A^T_TB_{T-1}
\end{bmatrix},}
\end{equation*}
where $A^t_s = \prod\limits_{k=s}^{t-1} A_k$ for any $s<t$ and $A_s^t = I$ for $s=t$.
\begin{proposition}
\label{prop:SDP_reformulation}
The optimization problem \eqref{eq:trace_lowerbound} is equivalent to the conic program
\begin{align*}
\max \ 
&\Tr(\mathbf{G^{\top} Q G W}) - 
\Tr\big(\mathbf{M}^{-1} F\big)
\\
\text{s.t.} \
&\Gamma \in \mathcal{M}, \ 
F \in \mathbb{S}_{+}^{mT},\mathbf{W} \in \mathbb{S}_{++}^{d(T+1)}, \ 
\mathbf{V} \in \mathbb{S}_{++}^{pT},
\\[1pt]
& E_{x_0} \in \mathbb{S}_{+}^{d}, \ 
E_{w_t} \in \mathbb{S}_{+}^{d}, \ 
E_{v_t} \in \mathbb{S}_{+}^{p}, \ \forall t \in [T],
\\[2pt]
&\begin{aligned}
&\Tr\!\left(\!W_0\! -\! 2E_{x_0}\!+\!\frac{\epsilon}{2}\Sigma^{-1}W_0\!\right)\! -\! \frac{\epsilon}{2}\log\left|E_{x_0}\! - \frac{\epsilon}{4}I\right| \leq \rho_{x_0}
\\
&- \Tr(\hat{X}_0) + \frac{\epsilon}{2}\left(d\log\left(\frac{\epsilon}{2}\right)-\log|\hat{X}_0\Sigma|\right),
\end{aligned}
\\[2pt]
&\begin{aligned}
&\Tr\left(\!W_{t+1}\!\! -\! 2E_{w_t}\!\!+\!\frac{\epsilon}{2}\Sigma^{-1}W_{t+1}\!\right)\!\! -\! \frac{\epsilon}{2}\!\log\!\left|E_{w_t}\!\! -\! \frac{\epsilon}{4}I\right|\! \!\leq\! \rho_{w_t}
\\
&- \Tr(\hat{W}_t) + \frac{\epsilon}{2}\left(d\log\left(\frac{\epsilon}{2}\right)-\log|\hat{W}_t\Sigma|\right),\ \forall t \in [T]
\end{aligned}
\\[2pt]
&\begin{aligned}
&\Tr\left(V_{t}\! -\! 2E_{v_t}\!+\!\frac{\epsilon}{2}\Sigma^{-1}V_{t}\right)\!-\! \frac{\epsilon}{2}\log\left|E_{v_t}\! -\! \frac{\epsilon}{4}I\right| \leq\rho_{v_t}
\\
&- \Tr(\hat{V}_t) + \frac{\epsilon}{2}\left(d\log\left(\frac{\epsilon}{2}\right)-\log|\hat{V}_t\Sigma|\right),\ \forall t \in [T]
\end{aligned}
\\[2pt]
& 
\begin{bmatrix}
\hat{X}_0^{1/2} W_0 \hat{X}_0^{1/2} +\frac{\epsilon^2}{16}I_n & E_{x_0} \\
E_{x_0} & I_n
\end{bmatrix} \succeq 0, 
\\[1pt]
& 
\begin{bmatrix}
\hat{W}_t^{1/2} W_{t+1} \hat{W}_t^{1/2} +\frac{\epsilon^2}{16}I_n & E_{w_t} \\
E_{w_t} & I_n
\end{bmatrix} \succeq 0,\ \forall t \in [T],
\\[1pt]
& 
\begin{bmatrix}
\hat{V}_t^{1/2} V_t \hat{V}_t^{1/2}  +\frac{\epsilon^2}{16}I_p& E_{v_t} \\
E_{v_t} & I_p
\end{bmatrix} \succeq 0, 
\ \forall t \in [T],
\\[1pt]
& 
\setlength\arraycolsep{.9pt}
\begin{bmatrix}
F & \mathbf{H^{\top} Q G W D^{\top}} + \Gamma/2
\\[1pt]
(\mathbf{H^{\top} Q G W D^{\top}} + \Gamma/2)^{\top} & \mathbf{D W D^{\top} + V}
\end{bmatrix}
\!\!\succeq\! 0\,.
% & W_0 \succeq \lambda_{\min}(\hat{X}_0) I, \quad 
% W_{t+1} \succeq \lambda_{\min}(\hat{W}_t) I, \quad 
% V_t \succeq \lambda_{\min}(\hat{V}_t) I, 
% \quad \forall t \in [T-1].
\end{align*}
\noindent
Here, $\mathcal{M}$ denotes the set of all strictly upper block triangular matrices of the form
\[
\Gamma = \begin{bmatrix}
0 & \Gamma_{1,2} & \Gamma_{1,3} & \cdots & \Gamma_{1,T} \\
 & 0 & \Gamma_{2,3} & \cdots & \Gamma_{2,T} \\
 &  & \ddots & \ddots & \vdots \\
&  &  & 0 & \Gamma_{T-1,T} \\
 & &&& 0
\end{bmatrix}
\in \mathbb{R}^{Tm \times Tp}\,,
\]
where $\Gamma_{t,s} \in \mathbb{R}^{m \times p}$ for every $t,s \in \mathbb{Z}$ with $1 \le t < s \le T$.
\end{proposition}
\begin{proof}
The proof relies on dualizing the inner minimization problem in \eqref{eq:trace_lowerbound}. Strong duality holds because the primal problem is feasible and involves only equality constraints, hence any feasible point is in fact a Slater point.
% Note that strong duality holds because the primal problem is trivially feasible and involves only equality constraints, 
% which implies that any feasible point is in fact a Slater point.

In the following we use $\Gamma \in \mathcal{M}$ to denote the Lagrange multiplier of the constraint $\bf U \in \mathcal{U}^{\text{lin}}$, 
which requires all blocks of $\bf U$ above the main diagonal to vanish. 
\noindent
The Lagrangian function of the inner minimization problem in \eqref{eq:trace_lowerbound} is
\begin{equation*}
\mathcal{L}(\mathbf{q}, \mathbf{U}, \Gamma) 
= \Tr\left(\mathbf{F}_{1} \mathbf{W}+\mathbf{F}_{2} \mathbf{V}\right)+\mathbf{q}^{\top} \mathbf{M} \mathbf{q} + \Tr(\mathbf{U}\Gamma^\top)
\end{equation*}
\noindent
Recall now that $\mathbf{R} \succ 0$ and $\mathbf{Q} \succeq 0$, and thus $\mathbf{M} \succ 0$. 
Consequently, $\mathcal{L}$ is minimized by $\mathbf{q}^{\star} = 0$ for any fixed $\bf U$ and $\Gamma$. 
In addition, the partial gradient of $\mathcal{L}$ with respect to $\bf U$ is given by
\[
\frac{\partial \mathcal{L}}{\partial \bf U}
= 2 \mathbf{M U D W D}^{\top} 
 + 2 \mathbf{M U V}
 + 2 \mathbf{H^\top QGWD^\top} + \Gamma\,.
\]
\noindent
Recall also that $\mathbf{V} \in \mathcal{G}_{V}$ is strictly positive, which implies that $\mathbf{D W D^{\top} }+ \mathbf{V} \succ 0$ is invertible. 
As we already know that $\mathbf{M}\succ 0$ is invertible as well, 
$\mathcal{L}$ is minimized by
\[
\mathbf{U}^{\star} 
= -\mathbf{M}^{-1} 
   (\mathbf{H^{\top} Q G W D^{\top}} + \Gamma/2)
   (\mathbf{D W D^{\top} + V})^{-1}
\]
for any fixed $\Gamma$.

\noindent
Substituting both $\bf q^{\star}$ and $\bf U^{\star}$ into $\mathcal{L}$ yields the dual objective function
\[
\begin{aligned}
&g(\Gamma)
= \mathcal{L}(\mathbf{q^{\star}, U}^{\star}, \Gamma)
= \Tr(\mathbf{G^{\top} Q G W})+
\\
&
\begin{multlined}
- \Tr\Big(
    (\mathbf{M}^{-1}
    (\mathbf{H^{\top} Q G W D^{\top}} + \Gamma/2)
    (\mathbf{D W D^{\top} + V})^{-1}
    \\
    (\mathbf{H^{\top} Q G W D^{\top}} + \Gamma/2)^{\top}
   \Big)\,.
\end{multlined}
\end{aligned}
\]

\noindent
The dual of the inner minimization problem in \eqref{eq:trace_lowerbound} is thus given by 
$\max_{\Gamma \in \mathcal{M}} g(\Gamma)$.
To linearize the dual function $g(\Gamma)$, we introduce the auxiliary variable $F \in \mathbb{S}_{+}^{mT}$ 
subject to the matrix inequality
\[
\begin{multlined}
F \succeq 
  (\mathbf{H^{\top} Q G W D^{\top}} + \Gamma/2)
    (\mathbf{D W D^{\top} + V})^{-1}
    \\
    (\mathbf{H^{\top} Q G W D^{\top}} + \Gamma/2)^{\top}\,.
\end{multlined}
\]
We can reformulate the previous inequality as an LMI using the Schur complement. Then, the dual problem rewrites as
\begin{equation}
\label{eq:inner_dual}
\begin{aligned}
\max \ \ &\Tr(\mathbf{G^{\top} Q G W}) - 
\Tr\big(\mathbf{M}^{-1} F\big)
\\[1pt]
\text{s.t.} \ \ & 
\Gamma \in \mathcal{M}, \ 
F \in \mathbb{S}_{+}^{mT},
\\[1pt]
& 
\!\!\setlength\arraycolsep{.9pt}
\begin{bmatrix}
F & \mathbf{H^{\top} Q G W D^{\top}} + \Gamma/2
\\[1pt]
(\mathbf{H^{\top} Q G W D^{\top}} + \Gamma/2)^{\top} & \mathbf{D W D^{\top} + V}
\end{bmatrix}
\!\!\succeq\!\!0\,.
\end{aligned}
\end{equation}
Substituting the strong dual \eqref{eq:inner_dual} in the inner minimization in \eqref{eq:trace_lowerbound}, we obtain
\begin{align*}
  \max \ \ &\Tr(\mathbf{G^{\top} Q G W}) - 
\Tr\big(\mathbf{M}^{-1} F\big)
\\[1pt]
\text{s.t.} \ \ & 
\Gamma \in \mathcal{M}, \ 
F \in \mathbb{S}_{+}^{mT},\mathbf{W} \in \mathbb{S}_{++}^{d(T+1)}, \ 
\mathbf{V} \in \mathbb{S}_{++}^{pT},
\\[1pt]
& 
\!\!\setlength\arraycolsep{.9pt}
\begin{bmatrix}
F & \mathbf{H^{\top} Q G W D^{\top}} + \Gamma/2
\\[1pt]
(\mathbf{H^{\top} Q G W D^{\top}} + \Gamma/2)^{\top} & \mathbf{D W D^{\top} + V}
\end{bmatrix}
\!\!\succeq\! 0,
\\[1pt]
&G_\epsilon(\hat{X}_0, W_0) \leq \rho_{x_0},
\\
&G_\epsilon(\hat{W}_{t+1}, W_{t+1}) \leq \rho_{w_t},\ G_\epsilon(\hat{V}_t, V_t) \leq \rho_{v_t} \ \forall t \in [T]\,.
\end{align*}
We conclude the proof by reformulating the entropy-regularized Gelbrich constraints. For simplicity, consider the constraint $G_\epsilon(\hat{X}_0, W_0) \leq \rho_{x_0}$; the reasoning applies \textit{mutatis mutandis} to the others. Consider an auxiliary variable $E_{x_0}\in\Symm^n_{+}$ subject to the matrix inequality $E_{x_0}^2 \preceq \hat{X}_0^{1/2} W_0 \hat{X}_0^{1/2} +\frac{\epsilon^2}{16}I_n$. This inequality can be recast as $E_{x_0} \preceq (\hat{X}_0^{1/2} W_0 \hat{X}_0^{1/2} +\frac{\epsilon^2}{16}I_n)^\frac{1}{2}$. With this auxiliary variable, we can equivalently express the nonlinear inequality constraint as
\begin{equation*}
\begin{multlined}
   \Tr\left(W_0 - 2E_{x_0}+\frac{\epsilon}{2}\Sigma^{-1}W_0\right) - \frac{\epsilon}{2}\log\left|E_{x_0} - \frac{\epsilon}{4}I\right| \leq \rho_{x_0}
   \\
   - \Tr(\hat{X}_0) + \frac{\epsilon}{2}\left(d\log\left(\frac{\epsilon}{2}\right)-\log|\hat{X}_0\Sigma|\right)\,.
\end{multlined}
\end{equation*}
If the constraint is satisfied with $E_{x_0}$ it will be also with the original variable since the trace and the log-determinant are monotone operators. Finally, with a Schur complement argument the inequality $E_{x_0}^2 \preceq \hat{X}_0^{1/2} W_0 \hat{X}_0^{1/2} +\frac{\epsilon^2}{16}I_n$ can be equivalently rewritten as
\begin{equation*}
    \begin{bmatrix}
\hat{X}_0^{1/2} W_0 \hat{X}_0^{1/2} +\frac{\epsilon^2}{16}I_n & E_{x_0} \\
E_{x_0} & I_n
\end{bmatrix} \succeq 0\,.
\end{equation*}
\end{proof}
% \input{sections/simulations}
% \addtolength{\textheight}{-12cm} 
% This command serves to balance the column lengths on the last page of the document manually. It shortens the textheight of the last page by a suitable amount.
% This command does not take effect until the next page so it should come on the page before the last. Make sure that you do not shorten the textheight too much.
\end{document}